\documentclass[11pt]{article}

\title{An Implementation of the Quantum Verification of Matrix Products Algorithm}
\author{Elton Pinto}
\date{}

\usepackage{amsmath}
\usepackage{amsthm}
\usepackage{amssymb}
\usepackage[utf8]{inputenc}
\usepackage[margin=1in]{geometry}
\usepackage{parskip}
\usepackage{graphicx}
\usepackage{hyperref}
\usepackage{algorithm}
\usepackage{algpseudocode}
\usepackage{caption}
\usepackage{subcaption}
\usepackage{cleveref}
\usepackage{enumitem}
\usepackage{booktabs}
\usepackage[braket, qm]{qcircuit}
\usepackage{listings}
\usepackage{inconsolata}

\crefname{section}{§}{§§}
\crefname{figure}{Fig.}{Fig.}
\crefname{algorithm}{Alg.}{Alg.}
\crefname{table}{Table}{Table}
\hypersetup{hidelinks}

\usepackage[
  url=false,
  doi=false,
  isbn=false,
  sortcites=true,
  sorting=ynt,
  citestyle=numeric,
  abbreviate=false,
  backend=biber
]{biblatex}
\theoremstyle{definition}

\newtheorem*{claim*}{Claim}

\theoremstyle{remark}

\setitemize{itemsep=1mm, parsep=0pt}

\begin{document}

\maketitle

\section{Introduction}

Quantum computing has experienced a recent surge in popularity given the
advancements in NISQ machines. Quantum algorithms are known to solve problems
like factoring numbers and simulating natural systems more efficiently than a
classical computer. Companies such as IBM, Google, and Rigetti have recently
built quantum computers that allow researchers to run quantum algorithms on
physical hardware. Given these developments, it has become important to survey
the feasibility of using quantum computing to solve large-scale problems.

Several studies have experimentally evaluated quantum algorithms on quantum
hardware. Mandviwalla et al. evaluated a 4-qubit implementation of Grover
search on the IBM-Q quantum processor \cite{mandviwalla2018implementing}.
Similarly, Acasiete et al. evaluated quantum random walks on graphs with 8 to
16 vertices \cite{acasiete2020implementation}. However, these studies do not
highlight the challenges involved when scaling to larger-sized circuits.

Extensive work has been carried out on developing quantum programming
frameworks. IBM has created Qiskit, a Python framework that supports prototyping
and executing quantum algorithms on simulators and quantum hardware. ORNL has
developed QCOR, a heterogenous classical-quantum framework that aims to use
quantum computers as accelerators akin to GPUs \cite{mintz2020qcor}. However,
little work has been done to evaluate the efficacy of using these frameworks to
develop quantum systems.

Our study aims to fill in these gaps by documenting the process of developing an
oracle for the Quantum Verification of Matrix Products (QVMP)
\cite{buhrman2005quantum}\cite{ambainis2002quantummatrix} algorithm. This can be
tricky to do given the non-deterministic nature of quantum programs and the
shortage of formal verification tooling in this space. We implement this
algorithm in Qiskit and demonstrate its functionality by running it on the Aer
simulator. We report a proof of oracle correctness, circuit metrics (gate count,
qubit count, circuit depth), transpilation times, and simulation times.

\section{Literature Review}

Quantum computing places an emphasis on thinking about how computation is
performed physically. It achieves speedup over classical algorithms by
exploiting physical phenomenon like entanglement to encode and perform
computation. The main promise of the field is that it can offer a non-trivial
speedup over classical computing. The algorithm that is frequently brought
up to demonstrate this promise is Shor’s algorithm. Shor’s
algorithm provides an exponential speedup over classical algorithms for
factoring numbers and is capable of cracking a subset of the widely-used RSA
encryption. It achieves this by using the Quantum Fourier Transform (QFT), the
quantum analogue of a Fourier transform, to perform phase estimation and
order-finding \cite{nielsen2000quantum}.

There are other algorithms like superdense coding, quantum key
distribution, and quantum simulation that have potential applications in
scientific simulation, machine learning, and cryptography. However, these
algorithms tend to require a very large number of qubits to be of practical use.
One of the major challenges in developing large-scale multi-qubit systems is
error-correction and noise. Before we reach the holy grail of  fault-tolerant
quantum systems, the field is currently attempting to make use of Noisy
Intermediate Quantum Computers (NISQ) to solve problems of important practical
use.

Grover's algorithm is a quantum search algorithm that provides an $O(\sqrt{N})$
algorithm for searching through unstructured data, which is a quadratic speed-up
over its classical counterpart. It does this by performing multiple Grover
iteration steps which constructively amplify states that correspond to search
results \cite{nielsen2000quantum}. A Grover iteration requires a user-defined
oracle to function correctly. The job of the oracle is to report back whether an
input (encoded as a quantum state) satisfies the search criteria. The core of
Grover search is farily straightforward to implement. The main challenge
here is efficiently encoding the oracle (which is typically described as
a classical decision function) as a quantum circuit. In this study, we try to
better understand this challenge by implementing QVMP which uses Grover search
(and therefore an oracle) as a sub-routine.

There are two popular algorithms for QVMP. The first algorithm, proposed by
Ambainis, Buhrman, Høyer, Karpinski, and Kurur, uses amplitude amplification
along with Grover search to look for a sub-matrix that doesn’t satisfy the
product \cite{ambainis2002quantummatrix}. This algorithm runs in
$O(n^{\frac{7}{4}})$ time and improves upon the optimal classical bound provided
by Freivalds \cite{freivalds1979fast}. The speedup is obtained because the
algorithm makes use of interference to arrive at a result in a smaller number of
iterations. However, metrics do not exist for the number of qubits required to
implement the oracles for the quantum search algorithms used, and the resources
required to carry out operations like multiplying sub-matrices.  Further, little
research has been done on evaluating the algorithm in a heterogenous
classical-quantum setup where quantum computers are used to accelerate certain
parts of the algorithm. There exists a 4-qubit physical implementation of Grover
search on IBM’s quantum processor \cite{mandviwalla2018implementing}. This
implementation tests IBM quantum computers on Grover’s algorithm to investigate
the impacts of different circuit and device attributes, and to highlight the
current capabilities of the system.  This study reports that current quantum
computers are able to solve the search problem on very small data sets. This is
similar to what our study intends to do, however, it does not investigate the
practicality of running algorithms that use Grover search and does not comment
on the composability of circuits and how it affects performance and results. The
second algorithm, proposed by Buhrman and Spalek, uses quantum random walks to
speed up the verification process and runs in $O(n^\frac{5}{3})$ time
\cite{buhrman2005quantum}.

Quantum random walks are analogous to classical walks, and have a number of
applications in quantum programming tasks. For example, they are used in solving
the element distinctness problem, in which the goal is to find if there exists a
set of M non-distinct elements in a domain of N elements
\cite{ambainis2007quantumwalk}. There have been attempts to run quantum random
walks on quantum hardware. Balu et al.  implemented an efficient physical
realization of a quantum random walk using $log_2(N)$ qubits to represent an
$N$-point lattice \cite{balu2018physical}.  Experimental evaluation was
carried out on the IBM-Q five-qubit processor. To overcome resource
requirements, they used a continuous time-limit quantum random walk
implementation. Acasiete et al. have implemented discrete-time quantum random
walks on IBM-Q, and were able to run quantum search based algorithms on graphs
with 8 and 16 vertices \cite{acasiete2020implementation}. They were able to
obtain results with 50\% fidelity, and claim that the results are more efficient
than equivalent classical algorithms.

There exists research on resource estimate quantification and benchmarking for
some quantum algorithms. Jaques et al. implemented Grover oracles for key
search on AES and LowMC encryption \cite{jaques2020implementing}. They lay
out a formal description of the oracle,  describe a reversible quantum-gate
implementation of the AES encryption-decryption algorithm, and estimate the
number of Clifford, T, and CNOT gates required for running circuits that can
crack AES-128, AES-192, and AES-256. The project uses Q\#, a quantum
programming language developed by Microsoft. The project reduces the circuit
depth of the Grover oracle by using internal parallelization, in which the
Grover search instance is run on disjoint subsets of the input domain.

A number of open-source frameworks exist for conducting quantum computing
research. IBM provides the Qiskit framework which lets researchers quickly
prototype and test algorithms on a simulator, and also run some workloads on a
quantum computer, the biggest one being the IBM-Q 16-qubit processor in
Melbourne. Fingerhuth et al. have compiled comparisons between Qiskit and other
frameworks like Quil, XACC, and ScaffCC \cite{fingerhuth2018open}. They comment
on the programming language choice, documentation, license, and general culture
around these communities. However, they do not compare these frameworks based on
their performance and ability to execute on quantum hardware. LaRose has
compared simulator performance and the quantum compiler of Qiskit and Rigetti
\cite{larose2019overview}.  The study does not report which algorithm was used
during performance evaluation. Instead, it qualifies a benchmark based on the
number of qubits used. ORNL has developed QCOR, which is a heterogenous
framework that aims to enable developers to use quantum computers as
accelerators, much like GPUs \cite{mintz2020qcor}. QCOR doesn’t support
amplitude amplification, quantum random walks, and basic circuits for performing
arithmetic as of now. Support will need to be added to facilitate
experimentation using the hybrid classical-quantum programming approach provided
by this framework. Salm et al. has worked on a NISQ analyzer that determines the
best quantum computer system to run a given workload based on the nature of the
quantum algorithm \cite{dustdar2020nisq}. They believe that this will improve
developer experience by obviating the need to understand complicated mathematics
to determine the best machine for running a particular quantum programming task.
None of the frameworks currently have a working implementation of quantum
verification of matrix products which we can use to perform benchmarking.

We believe that it is important to have estimates on how big of an input a
concrete implementation of an algorithm can process. We can use such evaluation
reports to gauge the current state of quantum computing and suggest areas which
need more improvement. Further, we can provide valuable feedback to library
authors about features that need to be added to facilitate productive quantum
algorithm research.

\section{Materials and Methods}

This section covers the major components of QVMP, the algorithm itself,
implementation details, and experimental setup.

\subsection{Grover search} \label{sec:grover_search}

Grover’s algorithm is a popular quantum search algorithm. Given an input space
of $N$ elements and an oracle $U_f$, Grover search can find $M$ solution
indices in $O(\sqrt{\frac{N}{M}})$ time. For simplicity, we assume that $N$ is
a power of 2. 

For $M = 1$ Grover search runs in $O(\sqrt{N})$ time, which is a quadratic
speedup over the classical algorithm for searching in an unstructured database
which takes $O(N)$ time. Therefore, Grover search offers a significant speedup.

The algorithm works by repeatedly applying a Grover operator $G$ to the initial
state $H^{\otimes n}\ket{0}^{\otimes n }$:

\begin{equation}
  G = (H^{\otimes n}(2\ket{0}\bra{0} - I)H^{\otimes n})U_f = (2\ket{\psi}\bra{\psi} - I)U_f
\end{equation}

It consists of the oracle $U_f$ and a phase shift operator
($2\ket{\psi}\bra{\psi} - I$) known as the diffuser. The specific
characteristics of $U_f$ are described in \cref{sec:grover_oracles}.

Each iteration can be geometrically viewed as a rotation of the state vector in
a plane spanned by the uniform superposition of solutions and non-solutions.
After the application of $U_f$, the diffuser rotates the state vector towards
the superposition of solutions. The number of such iterations can be shown to
be  $O(\sqrt{\frac{N}{M}})$. Therefore, in order to use Grover’s algorithm, one
needs to know the exact number of solutions $M$ in the search space.

The Grover operator circuit is summarized in \cref{fig:grover_operator_circuit}.

\begin{figure}[!h]
  \centering
  \includegraphics[width=0.7\textwidth]{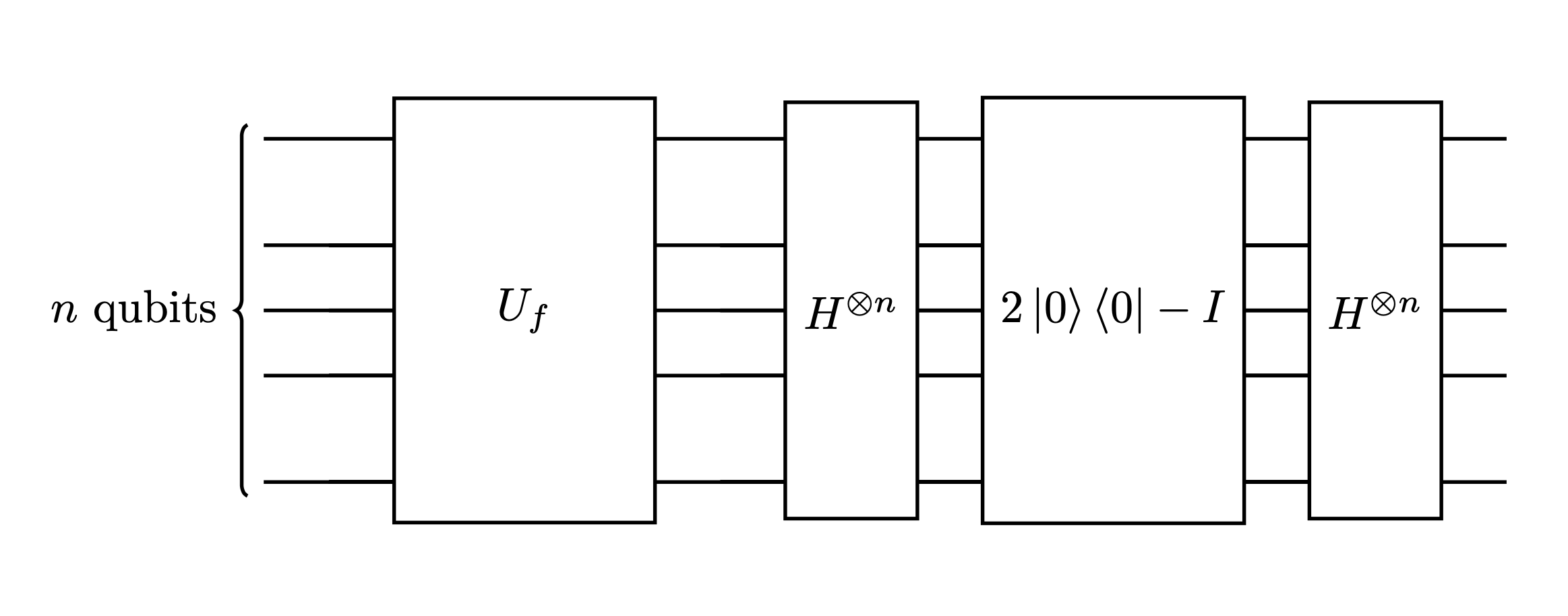}
  \caption{Grover operator circuit}
  \label{fig:grover_operator_circuit}
\end{figure}

\subsubsection{Grover Oracles} \label{sec:grover_oracles}

The oracle $U_f$ used in Grover’s algorithm can be viewed as black box that
knows how to recognize solutions in the search problem. Let us say we are given
a function $f$ which takes an input $x \in N$ and returns 1 if $x$ is a
solution to the search problem, and 0 otherwise. Then, the action of $U_f$ can
be written as:
\begin{equation}
  U_f\ket{x} \mapsto (-1)^{f(x)}\ket{x}
\end{equation}
Note how the oracle only applies a phase shift to solutions of the search
space. 

The above oracle is commonly implemented by encoding $f$ in a quantum circuit
that flips a target qubit $\ket{y}$ for all inputs $x$ that are a solution to the
search problem. We obtain the phase shift by initializing the target qubit to
the $\ket{-}$ state.

\subsection{Amplitude Amplification} \label{sec:amplitude_amplification}

Amplitude amplification is a generalization of the Grover operator $G$. Instead
of wrapping Hadamard gates $H$ around the diffuser, we now use an arbitrary
unitary $U$.
\begin{equation}
  U(2\ket{0}\bra{0} - I)U^{\dagger}O
\end{equation}
The oracle $O$ behaves in the same way as described in \cref{sec:grover_oracles}.
The unitary can be thought of as a quantum subroutine $A$ that performs a
series of quantum operations without making measurements.

\subsection{Quantum Verification of Matrix Products}

Given three square matrices $A$, $B$, and $C$ of size $n$, the verification of
matrix products (VMP) decides if $AB = C$. Freivalds describes a classical
algorithm which can run in $O(n^2)$ time.

In this paper, we implement the recursive Grover search based quantum VMP by
Ambainis et al. \cite{ambainis2002quantummatrix}. The algorithm proceeds by first
partitioning $B$ and $C$ into submatrices $B_i$ and $C_i$ of size $n \times
\sqrt{n}$ respectively. It is easy to observe that $AB = C$ iff $AB_i = C_i \
\forall i$. Now, perform amplitude amplification over the following subroutine:
pick a random vector x, classically compute $y = B_ix$ and $z = C_ix$, and
verify the product $Ay = z$. The verification is done using a Grover search
where the search space is the set of row indices and the oracle verifies if the
inner product between the row and the vector matches the output.

The verification oracle takes $O(n)$ time. Therefore, each Grover iteration
runs in $O(n^{\frac{3}{2}})$ time. We need to run $\sqrt{\frac{n}{\sqrt{n}}} =
n^{\frac{1}{4}}$ iterations of amplitude amplification. Therefore, the overall
running time of the algorithm is $O(n^{\frac{7}{4}})$.

The algorithm is summarized in \cref{alg:qvmp_grover}.

\begin{algorithm}
  \caption{Quantum VMP using Grover Search \cite{lanl2018quantum}}
  \label{alg:qvmp_grover}
  \textbf{Input: } $n \times n$ matrices $A, B, C$ \\
  \textbf{Output: } 1 if $AB = C$ and 0 otherwise \\
  \textbf{Procedure: }
  \begin{enumerate}
    \item Partition $B$ and $C$ into sub-matrices of size $n \times \sqrt{n}$
    \item 
      {
        Perform amplitude amplification for $n^{\frac{1}{4}}$ iterations using this subroutine:
        \begin{enumerate}
          \item Pick a random vector $x$ of size $\sqrt{n}$
          \item Classically compute $y = B_ix$ and $z = C_ix$
          \item Using Grover search with $\sqrt{n}$ iterations, find a row of
            index $j$ such that $(Ay \neq z)_j$
        \end{enumerate}
      }
    \item XOR the sub-results
  \end{enumerate}
\end{algorithm}

\subsection{Implementation}

We implement QVMP in Qiskit, a popular open-source quantum computing platform developed by IBM.
Qiskit uses Python as the host language and has a large libary of utilities that allow
developers to compose non-trivial circuits. It ships with a transpiler and
several backends capable of running circuits on simulators and quantum hardware.

We restrict our implementation to only support binary matrices. Remaining
details can be found in \cref{sec:analysis}.

\subsection{Experimental Setup}

The experiments were run on an AMD EPYC 7502 32-Core Processor, 2.5 GHz, 128
CPUs. The source code is available at \url{https://github.com/1ntEgr8/qvmp}.

\section{Implementation} \label{sec:analysis}

\begin{figure}
  \centering
  \scalebox{1.0}{\input{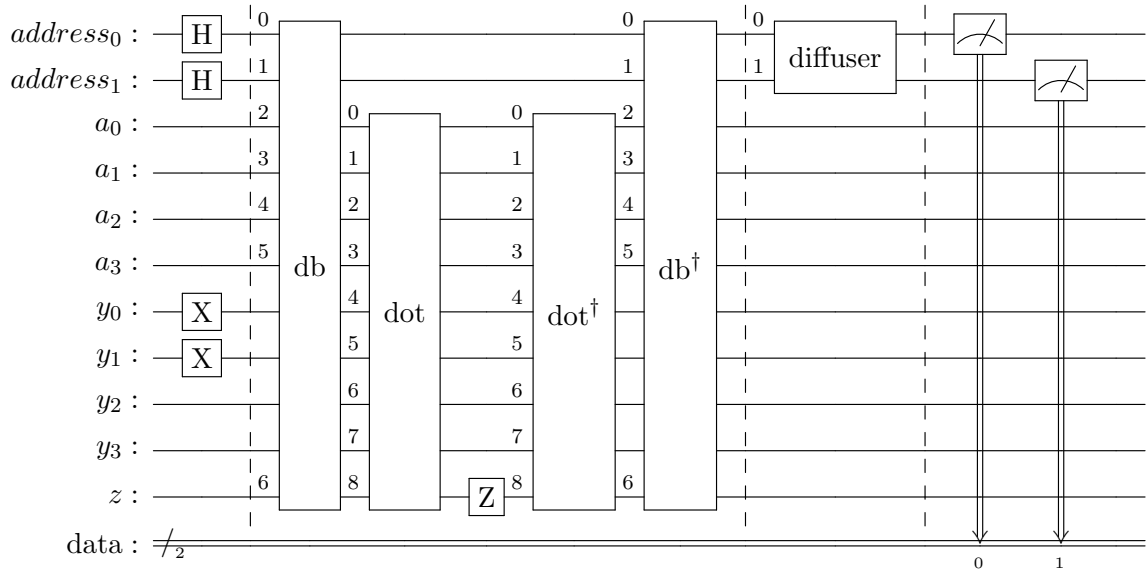}}
  \caption{QVMP circuit for a $4 \times 4$ matrix $A$ performing one iteration}
  \label{fig:qvmp_oracle_4x4}
\end{figure}

This section describes and analyzes the implementation of the quantum portion of
the algorithm (steps 2 and 2.3 of \cref{alg:qvmp_grover}). 

\cref{fig:qvmp_oracle_4x4} shows an example circuit that performs one
iteration of the Grover operator used in step 2.3. The sub-circuits used are
Quantum Read-Only Memory (QROM) \cite{babbush2018encoding} (denoted by $db$), out-of-place
inner product (denoted by $dot$), and diffuser.

Amplitude amplification (step 2) is carried out by first running Grover
search (step 2.3) and then appending a diffuser circuit. Describing it in terms
of the notation developed in \cref{sec:amplitude_amplification}, $U$ would be
the Grover search circuit.

To perform Grover search, we need to implement an oracle. The QVMP oracle is a
blackbox that checks if $(Ay - z)_j \neq 0$ where $A$, $y$, $z$, and $j$ are as
defined in \cref{alg:qvmp_grover}. Qiskit offers the \verb+classical_function+
decorator which allows you to describe Grover oracles classically. However, the
programs are restricted to using boolean operators and do not have support for
higher-level primitives like matrices, vectors, and related operations (which we
need). This means that we needed to hand-code the oracle from scratch.

The QVMP oracle needs to perform the following steps:
\begin{itemize}[itemsep=1mm, parsep=0pt]
  \item Encode the inputs $A$, $y$, and $z$
  \item Calculate the inner product $Ay$ and for each row index $j$ compute $(Ay
    \neq z)_j$
  \item Mark a qubit if the constraint is satisfied (thereby obtaining a marking
    oracle)
  \item Convert the marking oracle into a phase oracle
  \item Uncompute on the ancilla qubits
\end{itemize}
We encode $A$ and $z$ using QROM. $y$ is encoded using a bit-to-qubit encoding.
The inner product is computed using an out-of-place inner product circuit. We
use a $Z$ gate to perform the phase-flip on the marked qubit, which correctly
converts our oracle from marking to phase (see \cref{proof:oracle_correctness}).
Finally, we uncompute the inner-product and QROM circuits to return the ancilla
qubits to their original state.

\subsection{QROM}

\begin{figure}
  \centering
  \input{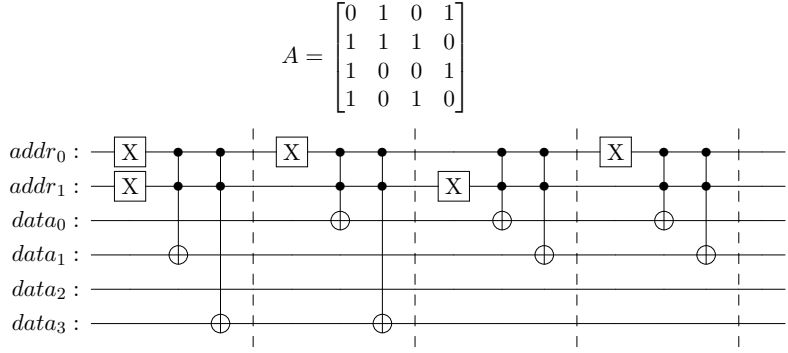}
  \caption{QROM encoding of a $4 \times 4$ matrix $A$}
  \label{fig:qrom_4x4}
\end{figure}

QROM can encode an $n \times m$ binary matrix using $m + \log_2(n)$ qubits. This
is quadratic improvement over a naive bit-to-qubit encoding where each binary
value is assigned to a qubit (resulting in $O(n^2)$ space-complexity). 

QROM takes in a row index encoded as a binary value using $log_2(n)$ address
qubits encodes the data for the row in $m$ data qubits.  This is done by
partitioning the circuit into $log_2(n)$ chunks, where each chunk corresponds to
an address. Using multi-controlled NOT gates it selectively updates the data
qubits corresponding to the input address with the corresponding values.

Since this is a quantum circuit, we can use superposition to pass in multiple
row-indices through the address qubits and read multiple rows. The caveat is, of
course, that we can only measure one of the rows. Nonetheless, we can still
exploit this feature of QROMs to perform the same computation on multiple rows
and boost the probability of measure a particular row through amplitude
amplification for example.

\cref{fig:qrom_4x4} shows a QROM circuit encoding a $4 \times 4$ matrix $A$.

\subsubsection{QROM in QVMP}

We use QROM to encode $A$ and $z$ (concatenated as an $n \times (m+1)$ matrix.
There is, however, the question of why we can't encode $y$ using QROM as well.
The answer is that we will not be able to perform the inner product between
$A_j$ and $y$ using such a scheme. This is because to perform the inner product
we need a single row of $A$ but \textbf{all} the rows in the vector $y$ because
of the way matrix multiplication works. $z$ can be encoded using QROM since we
only need the value of $z_j$.

\subsection{Inner Product (out-of-place)}

\begin{figure}
  \centering
  \input{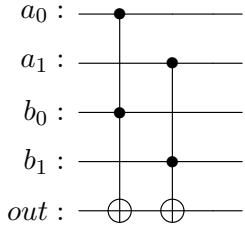}
  \caption{Inner product circuit for $2$-D vectors}
  \label{fig:inner_product}
\end{figure}

The inner product circuit, as the name suggests, computes the inner product
between two binary vectors encoded using bit-to-qubit encoding. Addition and
multiplication are as those typically defined for the field $F_2$. The circuit
consists of a series of Tofolli gates controlled on the corresponding inputs
with a single target qubit different from the input qubits (hence out-of-place).

\cref{fig:inner_product} shows an inner-product circuit between $2$-D vectors.

\subsubsection{Inner Product in QVMP}

We use the inner product circuit to compute $(Ay)_j$. Note that we are using $z$
as the target qubit. This serves the purpose of both calculating $(Ay)_j$ and
comparing it to $z$. The output of $z$ will be $0$ if they don't match and $1$
if they do. This is precisely why performing a $Z$ gate on $z$ is sufficient to
convert our oracle from marking to phase.

\subsection{Diffuser}

\begin{figure}
  \centering
  \input{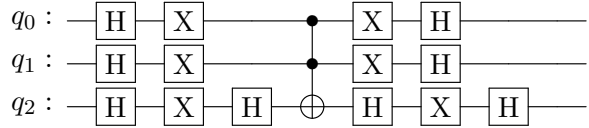}
  \caption{Diffuser running on three qubits}
  \label{fig:diffuser_3}
\end{figure}

The diffuser is implemented exactly as described in
\cref{sec:grover_search}. \cref{fig:diffuser_3} shows a diffuser running on
three qubits.

\subsection{Oracle correctness} \label{proof:oracle_correctness}

\begin{claim*}
  Given a superposition of row indices $addr$, an $n \times m$ binary matrix $A$,
  and $m$-dimensional binary vectors $y$ and $z$, the QVMP oracle will flip the
  phase of the row indices $j$ which satisfy $(Ay \neq z)_j$ and leave the ancilla
  qubits unchanged.
\end{claim*}
\begin{proof}
  Let $a_1, a_2, \ldots, a_k$ be row indices selected from the set $\{0, \ldots,
  n-1\}$ for some $0 \leq k \leq n$.

  The initial state of the circuit (after encoding $y$) is
  \begin{equation*}
    \psi_0 = (\alpha_1\ket{a_1} + \alpha_2\ket{a_2} + \ldots +
    \alpha_k\ket{a_k})\ket{0}^{\otimes m}\ket{y}\ket{0}
  \end{equation*}
  where $\alpha_i \in \mathbb{C}$ are the amplitudes.

  Applying the oracle (using the notation of \cref{fig:qvmp_oracle_4x4}),
  \begin{align*}
    \psi_0 & \xrightarrow{db}
      (\alpha_1\ket{a_1}\ket{A_{a_1}}\ket{z_{a_1}} + \ldots +
       \alpha_k\ket{a_k}\ket{A_{a_k}}\ket{z_{a_k}})\ket{y} \\
    & \xrightarrow{dot}
      (\alpha_1\ket{a_1}\ket{A_{a_1}}\ket{z_{a_1} \oplus (Ay)_{a_1}} + \ldots +
       \alpha_k\ket{a_k}\ket{A_{a_k}}\ket{z_{a_k} \oplus (Ay)_{a_k}})\ket{y} \\
    & \xrightarrow{Z}
      ((-1)^{z_{a_1} \oplus
       (Ay)_{a_1}}\alpha_1\ket{a_1}\ket{A_{a_1}}\ket{z_{a_1} \oplus (Ay)_{a_1}} + \ldots +
       (-1)^{z_{a_k} \oplus
       (Ay)_{a_k}}\alpha_k\ket{a_k}\ket{A_{a_k}}\ket{z_{a_k} \oplus (Ay)_{a_k}})\ket{y} \\
    & \xrightarrow{dot^\dagger}
      ((-1)^{z_{a_1} \oplus
       (Ay)_{a_1}}\alpha_1\ket{a_1}\ket{A_{a_1}}\ket{z_{a_1}} + \ldots +
       (-1)^{z_{a_k} \oplus
       (Ay)_{a_k}}\alpha_k\ket{a_k}\ket{A_{a_k}}\ket{z_{a_k}})\ket{y} \\
    & \xrightarrow{db^\dagger}
      ((-1)^{z_{a_1} \oplus
       (Ay)_{a_1}}\alpha_1\ket{a_1} + \ldots +
       (-1)^{z_{a_k} \oplus
       (Ay)_{a_k}}\alpha_k\ket{a_k})\ket{0}^{\otimes m}\ket{y}\ket{0}
  \end{align*}
  $(-1)^{z_{a_i} \oplus (Ay)_{a_k}}$ equals $-1$ if $(Ay \neq z)_{a_k}$. It is 1
  otherwise.

  This proves the claim.
\end{proof}

\section{Results} \label{sec:results}

QVMP specifies that we run $N = n^{\frac{1}{4}}$
iterations of the Grover operator regardless of the number of solutions $m$. The
optimal number of iterations $N_{\text{optimal}} \approx
\frac{\pi}{4}\sqrt{\frac{n}{m}} \leq N,\, \forall \, n$. We can use either
number of iterations to perform QVMP. The former requires performing amplitude
amplification and while the latter doesn't. We explicitly mention which one we
are using when discussing our results.

The simulations were carried out using the Aer backend provided by Qiskit. Aer
supports different simulation methods, two of which are statevector and matrix
product state (MPS). The statevector method performs a dense statevector
simulation of the quantum state. The memory requirements for this method scale
exponentially with the number of qubits. We were able to simulate circuits
containing up to 32 qubits with this method before hitting memory limits. The MPS
method allows for more efficient operations on circuits with relatively low
entanglement. We were able to simulate the entirety of our desired input space
using this method.

\subsection{Functionality}

\begin{figure}[h!]
  \centering
  \begin{subfigure}{0.48\textwidth}
    \centering
    \includegraphics[width=\textwidth]{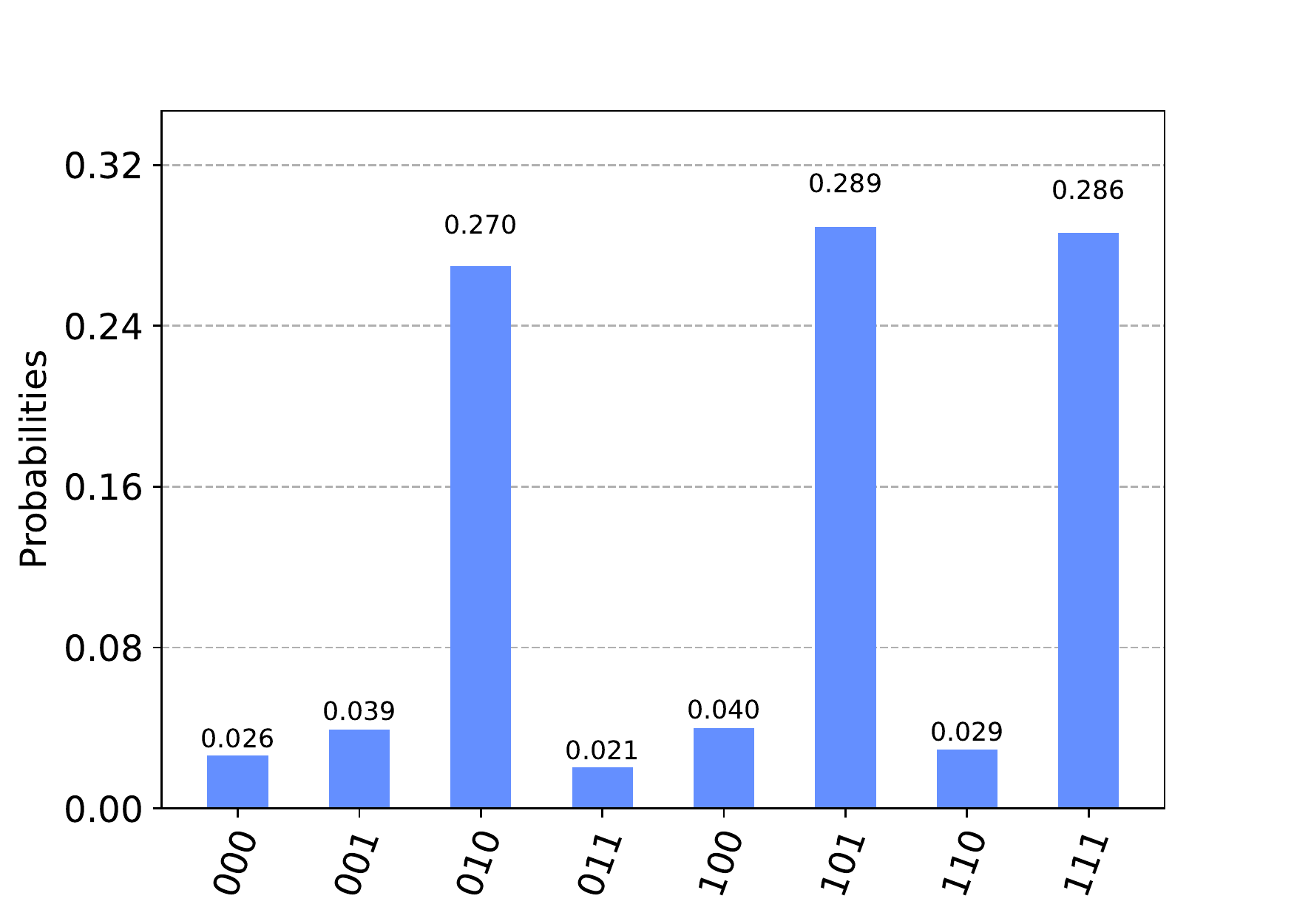}
    \caption{Using $N_{\text{optimal}}$}
  \end{subfigure}
  \begin{subfigure}{0.48\textwidth}
    \centering
    \includegraphics[width=\textwidth]{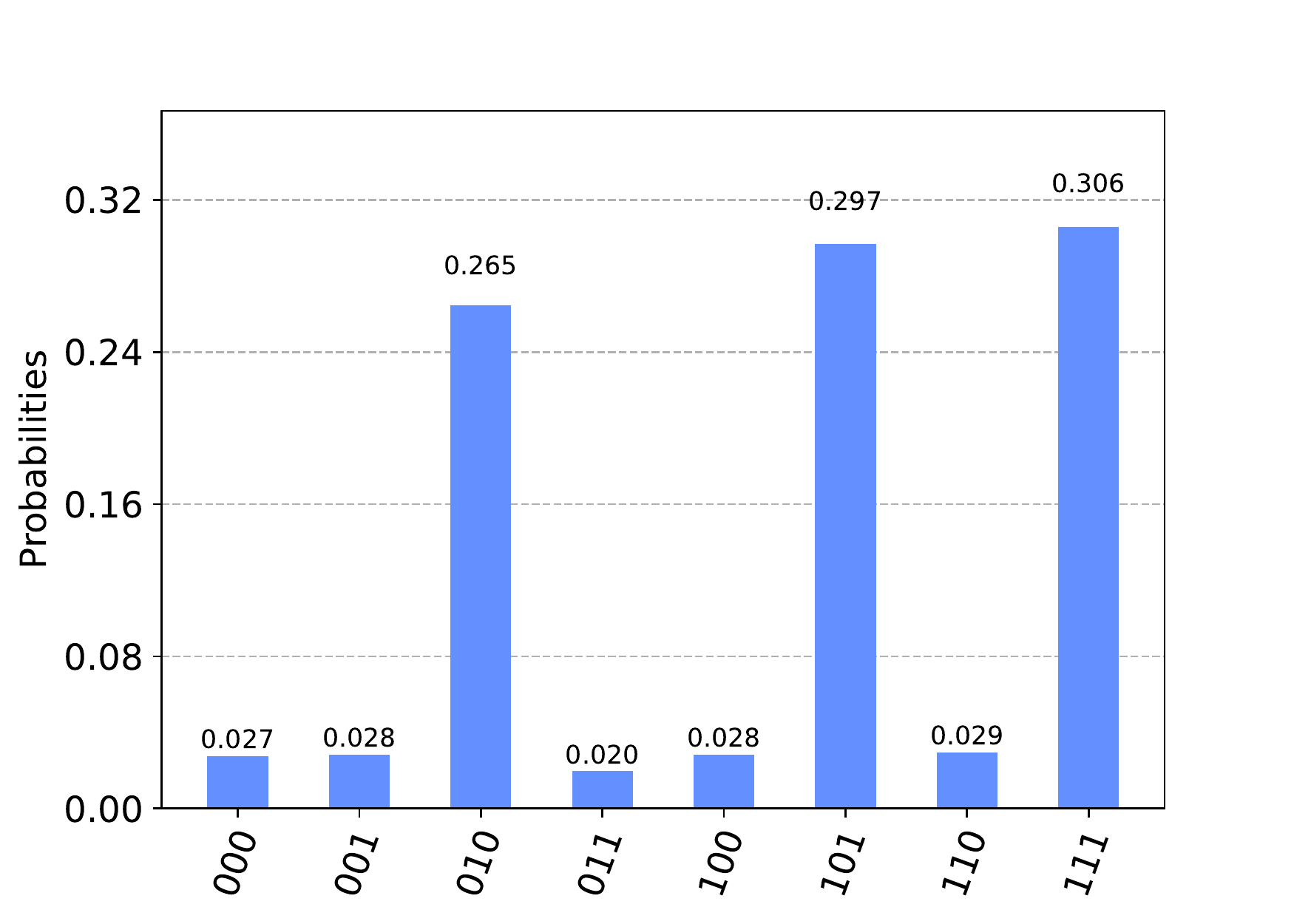}
    \caption{Using $N$}
  \end{subfigure}
  \caption{Histogram showing the probability distribution of measuring the
  address qubits of the QVMP circuit for an $8 \times 8$ matrix with $(Ay \neq
  z)_j$ for $j \in \{2, 5, 7\}$}
  \label{fig:qvmp_functionality_found}
\end{figure}

\cref{fig:qvmp_functionality_found} demonstrates that our implementation is able
to find the row indices under both $N$ and $N_{\text{optimal}}$ iterations.

\begin{figure}[h!]
  \centering
  \begin{subfigure}{0.48\textwidth}
    \includegraphics[width=\textwidth]{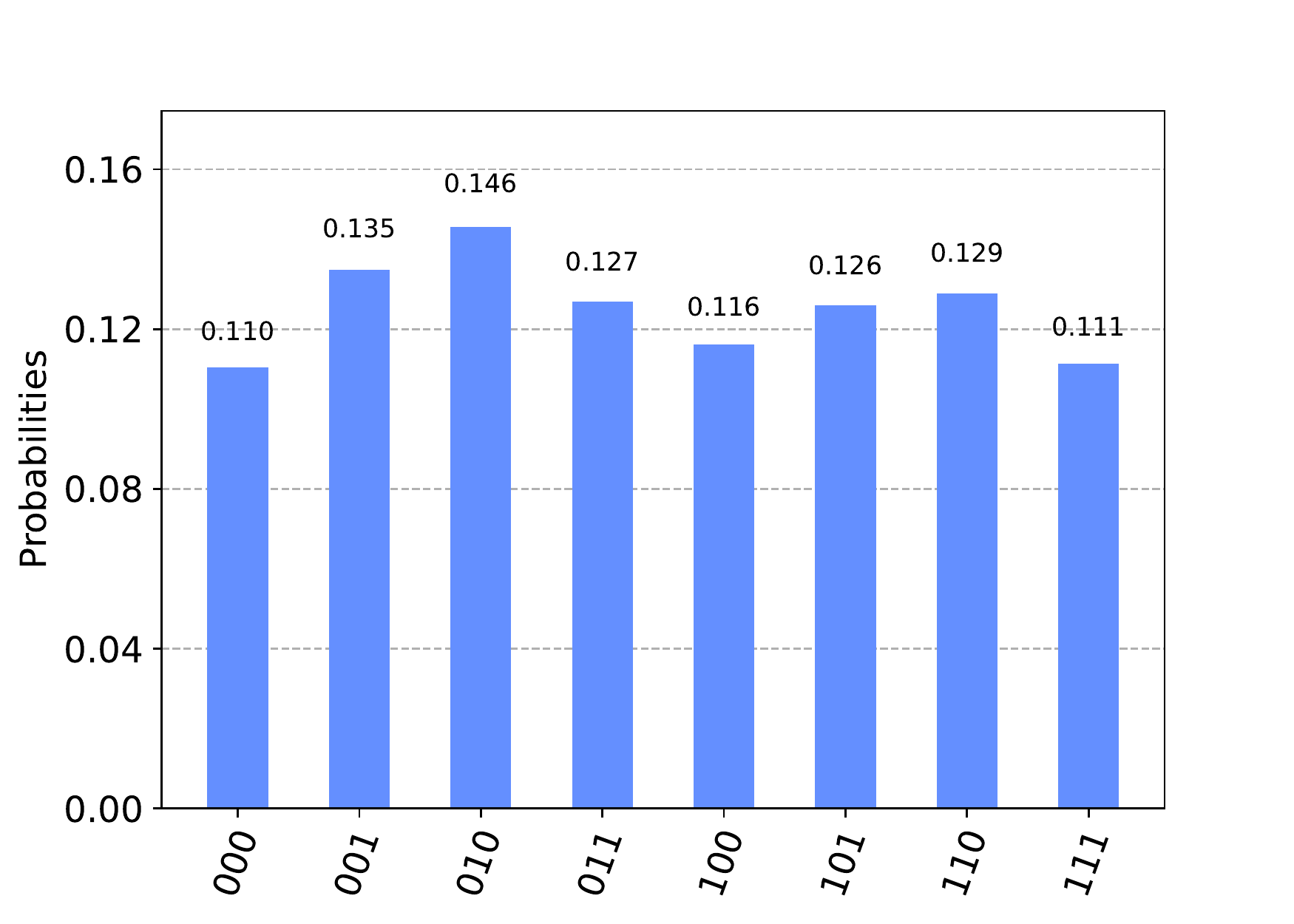}
    \caption{Using $N_{\text{optimal}}$}
  \end{subfigure}
  \begin{subfigure}{0.48\textwidth}
    \centering
    \includegraphics[width=\textwidth]{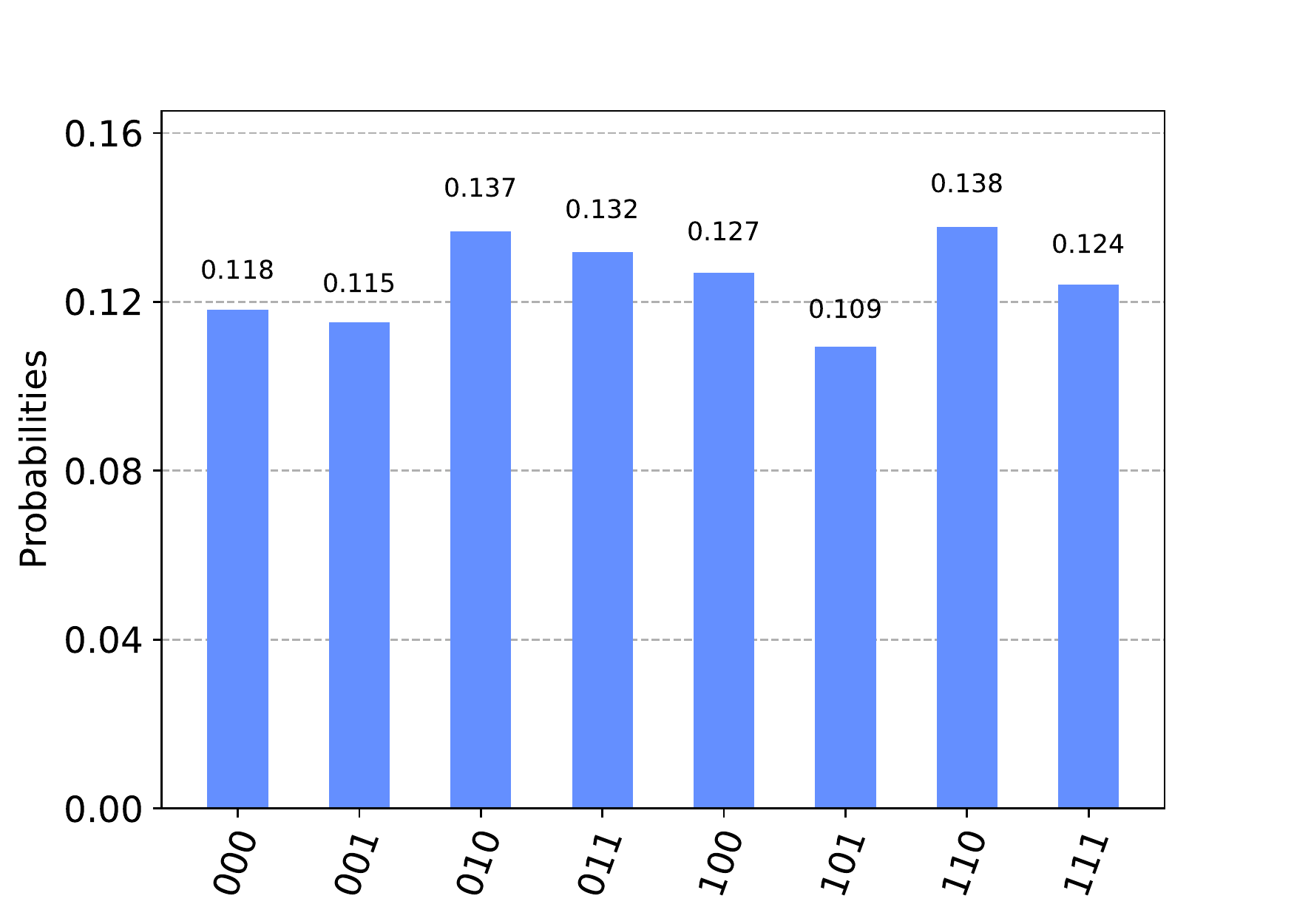}
    \caption{Using $N$}
  \end{subfigure}
  \caption{Histogram showing the probability distribution of measuring the
  address qubits of the QVMP circuit for an $8 \times 8$ matrix with no $j$
  satisfying $(Ay \neq z)_j$}
  \label{fig:qvmp_functionality_none}
\end{figure}

\cref{fig:qvmp_functionality_none} demonstrates the case when no solutions are
present. There is an equal probability of measuring any of the row indices.

\begin{figure}[h!]
  \centering
  \begin{subfigure}{0.48\textwidth}
    \centering
    \includegraphics[width=\textwidth]{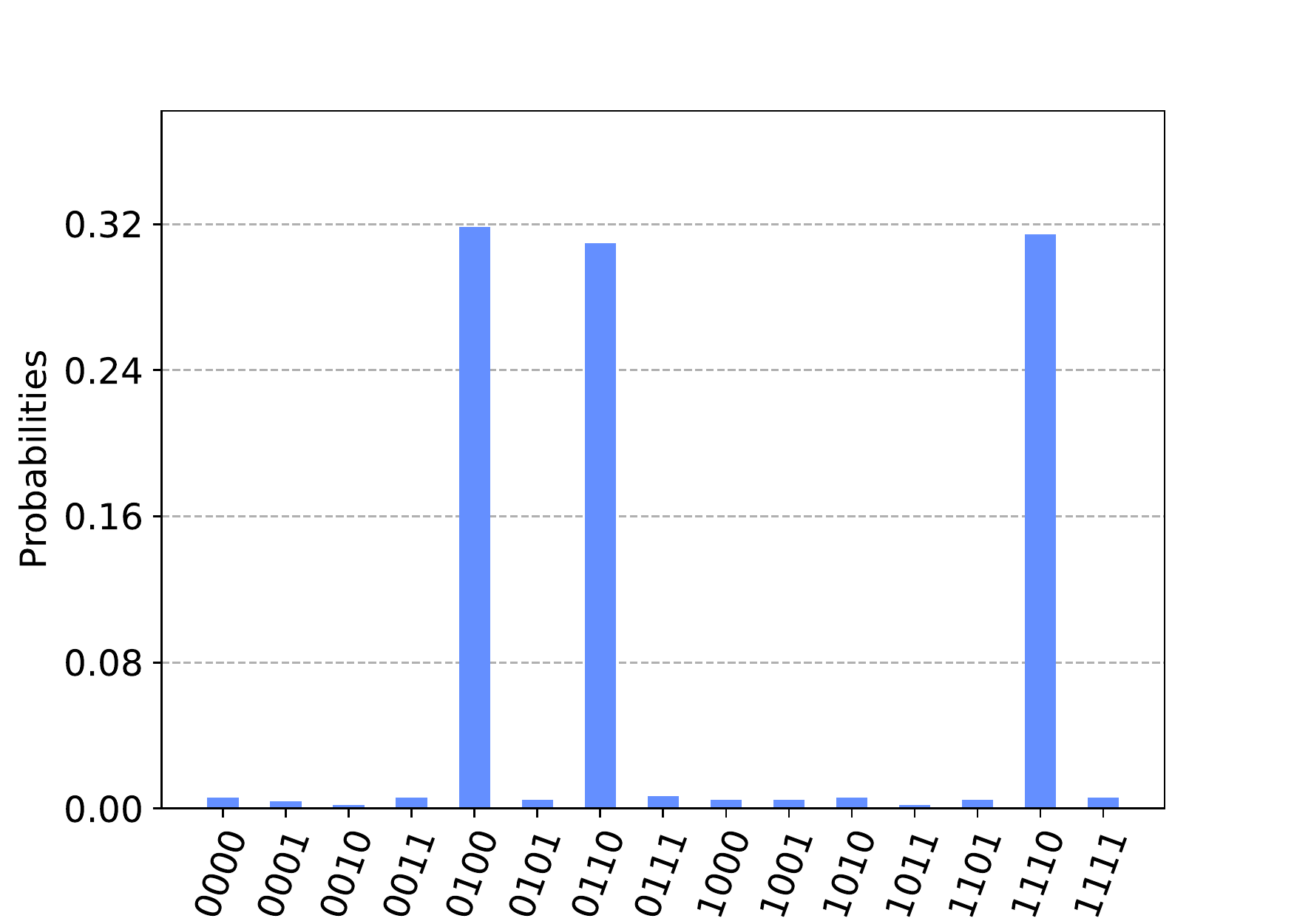}
    \caption{Using $N_{\text{optimal}} = 1$}
  \end{subfigure}
  \begin{subfigure}{0.48\textwidth}
    \centering
    \includegraphics[width=\textwidth]{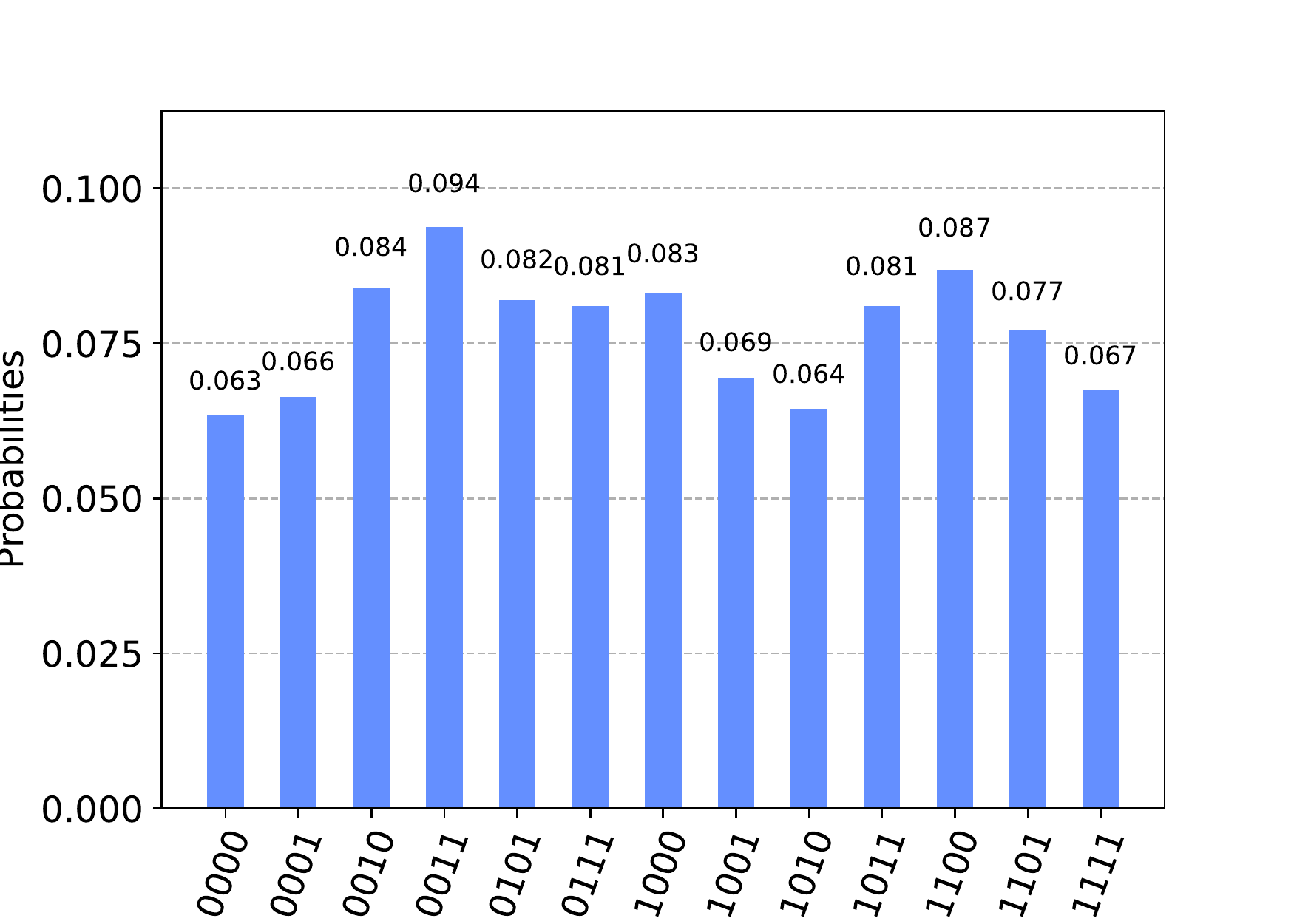}
    \caption{Using $N = 4$}
  \end{subfigure}
  \caption{\textbf{Failure case} Histogram showing the probability distribution
  of measuring the address qubits of the QVMP circuit for a $16 \times 16$
  matrix with $(Ay \neq z)_j$ $j \in \{1, 4, 5\}$. $N \neq \lfloor
  k*\text{period} \rfloor$}
  \label{fig:qvmp_functionality_pfound_unknown__1}
\end{figure}

\begin{figure}[h!]
  \centering
  \begin{subfigure}{0.48\textwidth}
    \includegraphics[width=\textwidth]{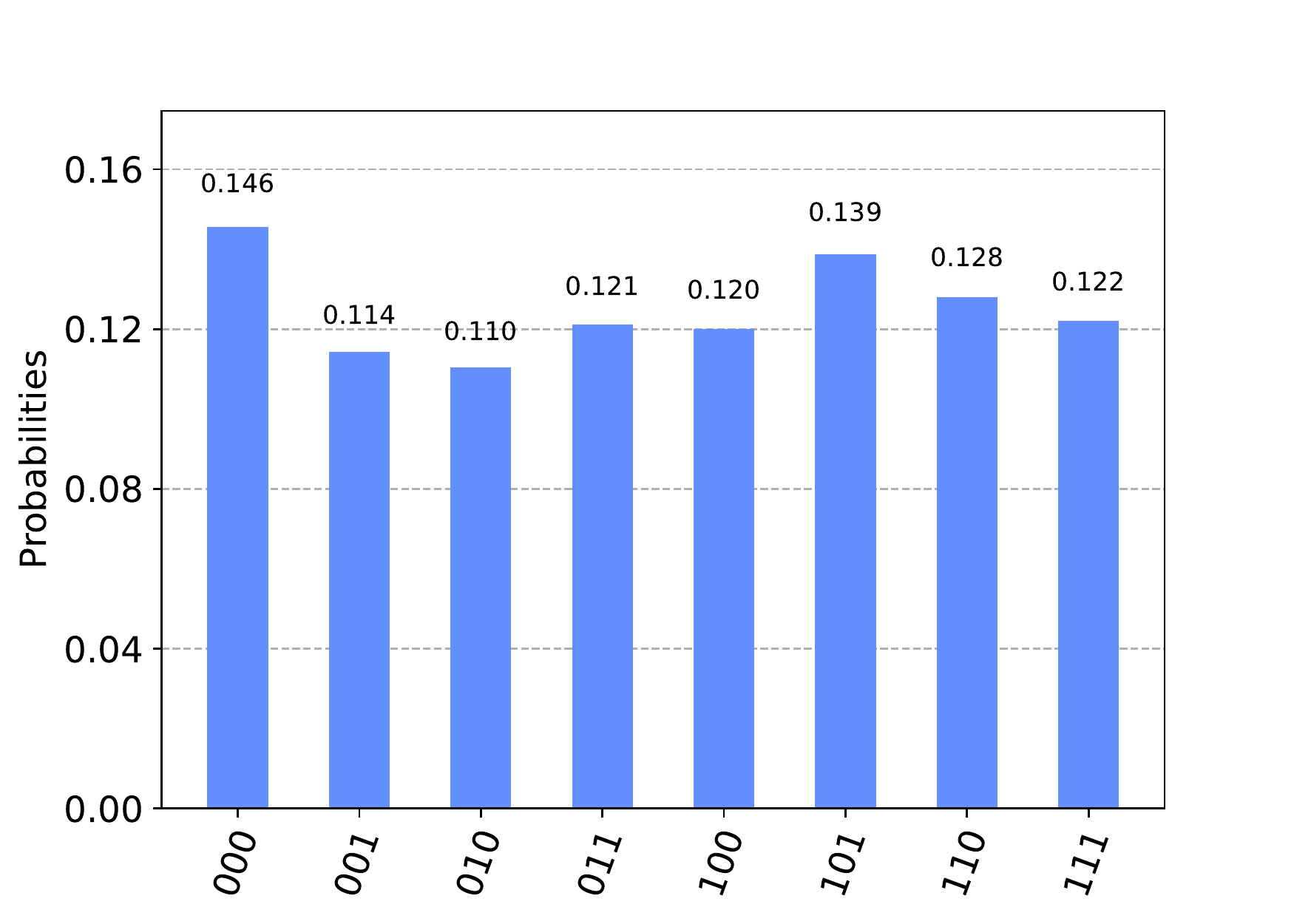}
    \caption{Using $N_{\text{optimal}} = 0$}
  \end{subfigure}
  \begin{subfigure}{0.48\textwidth}
    \includegraphics[width=\textwidth]{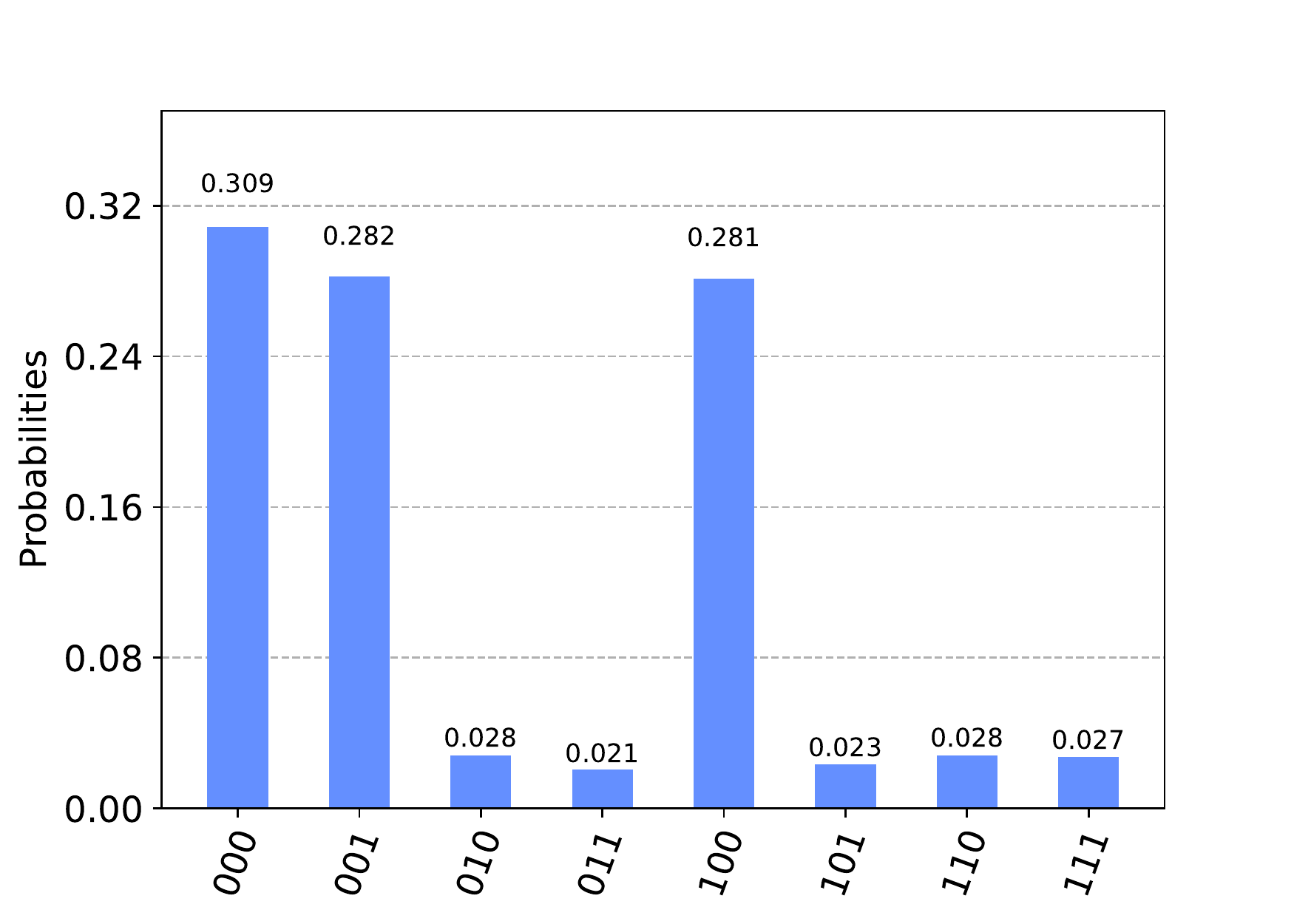}
    \caption{Using $N = 2$}
  \end{subfigure}
  \caption{\textbf{Failure case} Histogram showing the probability distribution
  of measuring the address qubits of the QVMP circuit for an $8 \times 8$ matrix
  with $(Ay \neq z)_j$ for $j \in \{2, 3, 5, 6, 7\}$. $N_{\text{optimal}} = 0$}
  \label{fig:qvmp_functionality_pfound_unknown__2}
\end{figure}

As the number of iterations increases, the amplitudes start oscillating around
the optimal value.  We define the period to be the number of oscillations
between optimals.  \cref{fig:qvmp_functionality_pfound_unknown__1} and
\cref{fig:qvmp_functionality_pfound_unknown__2} demonstrate cases when the
algorithm fails to find a solution even when one exists.  The former occurs when
$N$ is not equal to $\lfloor k * \text{period} \rfloor$ for some integer $k$.
For example, in \cref{fig:qvmp_functionality_pfound_unknown__1}, the period
$\approx 1.313$, $N = 4$, and the sequence of periods is $1, 2, 3, 5, \ldots$.
We can handle this by hypertuning the number of iterations to be close to the
optimal.  The latter occurs when the number of optimal Grover iterations rounds
down to zero. We can handle this by running the algorithm on the dual problem,
i.e., finding $j$ where $(Ay = z)_j$.

\subsection{Circuit Metrics} \label{sec:circuit_metrics}

\begin{table*}
    \centering
    \begin{subtable}{\textwidth}
      \resizebox{\textwidth}{!}{%
        \begin{tabular}{|l|c||c|c|c|c|c|c|c|c||c|c|c|c|}
          \hline
          Dimension & Mismatches & ccx &      cx &   x &     h & z &      u1 &    u2 & u3 & Grover iterations &   Depth & Qubits & Total gates \\
          \hline
           (4, 4) &          1 &  30 &       1 &  11 &     2 & 1 &       2 &     2 &  0 &                 1 &      44 &     11 &          49 \\
          (16, 4) &          2 &  16 &    6220 &  68 &   176 & 2 &    5216 &  1564 &  3 &                 2 &   10585 &     13 &       13265 \\
          (16, 8) &          2 &  32 &   10108 &  69 &   284 & 2 &    8460 &  2536 &  3 &                 2 &   16993 &     21 &       21494 \\
          (32, 4) &          2 &  24 &   42060 & 204 &   461 & 3 &   42037 &   504 &  6 &                 3 &   69510 &     14 &       85299 \\
          (32, 8) &          1 &  64 &   91408 & 272 &   997 & 4 &   91377 &  1056 &  8 &                 4 &  150375 &     22 &      185186 \\
         (32, 32) &          3 & 128 &  185912 & 151 &  2025 & 2 &  185897 &  2052 &  4 &                 2 &  303901 &     70 &      376171 \\
          (64, 8) &          3 &  48 &  300324 & 401 &  1602 & 3 &  300315 &  1632 &  4 &                 3 &  497172 &     23 &      604329 \\
         (64, 16) &          2 & 128 &  786960 & 538 &  4190 & 4 &  786948 &  4232 &  5 &                 4 & 1300351 &     39 &     1583005 \\
         (64, 64) &          3 & 384 & 2329596 & 432 & 12396 & 3 & 2329587 & 12426 &  4 &                 3 & 3843672 &    135 &     4684828 \\
          \hline
        \end{tabular}}
      \caption{MPS}
      \label{table:circuit_metrics_mps}
    \end{subtable}
    \begin{subtable}{\textwidth}
      \resizebox{\textwidth}{!}{%
        \begin{tabular}{|l|c||c|c|c|c|c|c|c|c||c|c|c|c|}
          \hline
          Dimension & Mismatches & ccx &      cx &   x &     h & z &      u1 &    u2 & u3 & Grover iterations &   Depth & Qubits & Total gates \\
          \hline
          (4, 4) &          1 &  30 &  1 &  11 & 2 & 1 &  2 &  2 &  0 &                 1 &    44 &     11 &          49 \\
         (16, 4) &          2 &  16 &  0 &  75 & 4 & 2 &  3 & 12 &  1 &                 2 &   261 &     13 &         113 \\
         (16, 8) &          2 &  32 &  0 &  76 & 4 & 2 &  3 & 12 &  1 &                 2 &   385 &     21 &         130 \\
         (32, 4) &          2 &  24 &  0 & 208 & 5 & 3 &  4 & 24 &  2 &                 3 &   684 &     14 &         270 \\
         (64, 8) &          3 &  48 &  0 & 405 & 6 & 3 &  4 & 30 &  2 &                 3 &  2040 &     23 &         498 \\
         (64, 8) &          1 &  96 &  0 & 801 & 6 & 6 &  7 & 60 &  5 &                 6 &  4077 &     23 &         981 \\
          \hline
        \end{tabular}}
      \caption{Statevector}
      \label{table:circuit_metrics_statevector_cpu}
    \end{subtable}
    \caption{Circuit metrics for MPS and statevector simulation methods on
    select dimensions. Depth and total gates were measured after transpilation.}
    \label{table:circuit_metrics}
  \end{table*}

\cref{table:circuit_metrics} reports circuit metrics (gate count, circuit depth,
qubit count) measured when using the MPS and statevector methods for the Grover
search portion of QVMP (step 2.3 in \cref{alg:qvmp_grover}. The metrics are
reported on circuits that used $N_{\text{optimal}}$ iterations.

The qubit count is precisely $m + \log_2(n) + n + 1$ for an $n \times m$ matrix.
In other words, it is linear in the size of the input. Circuit depth is an
important metric because it impacts both transpilation time and fidelity of the
output.

\begin{figure}[h!]
  \begin{subfigure}{0.48\textwidth}
    \includegraphics[width=\textwidth]{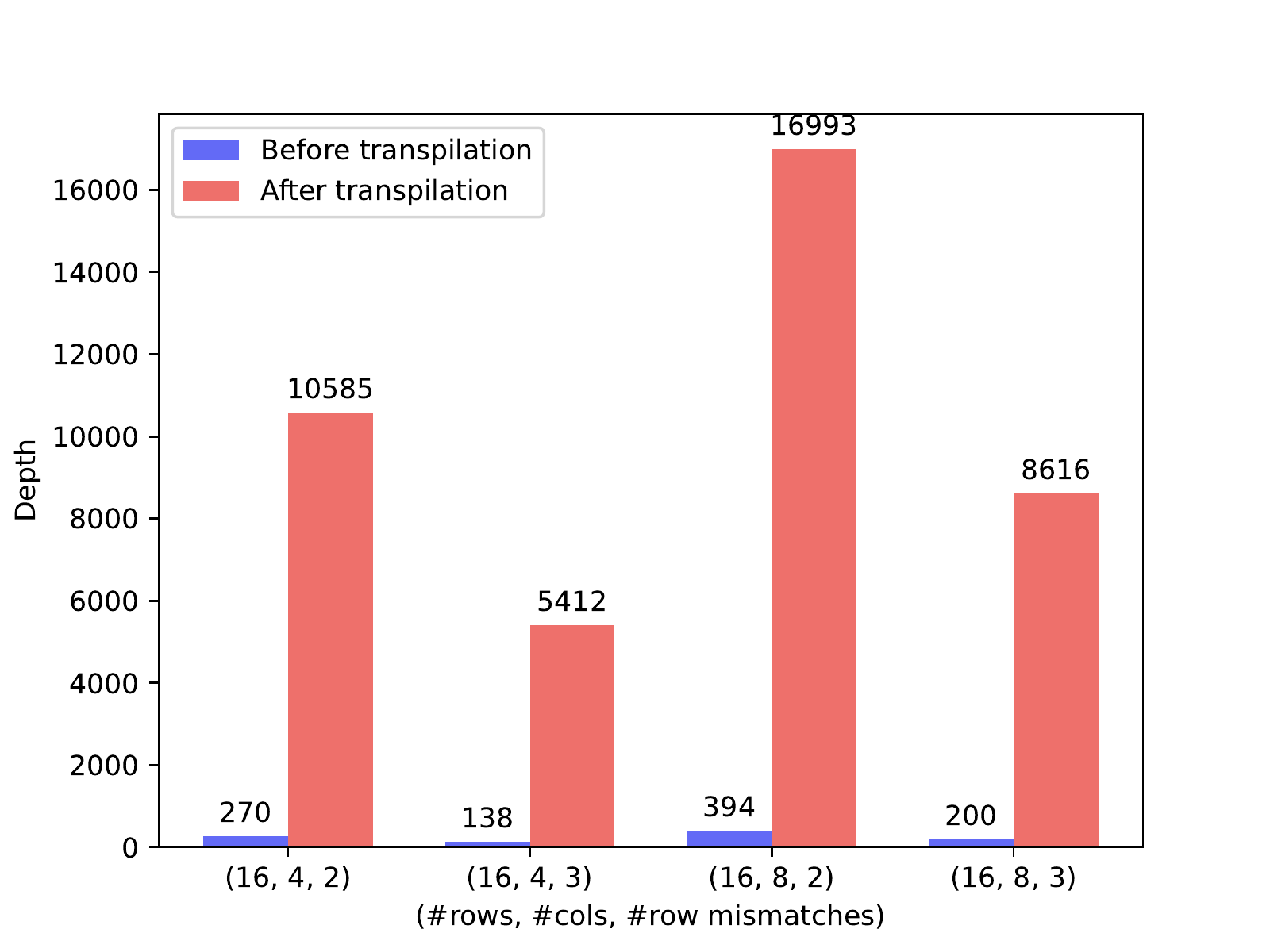}
    \caption{MPS}
  \end{subfigure}
  \begin{subfigure}{0.48\textwidth}
    \includegraphics[width=\textwidth]{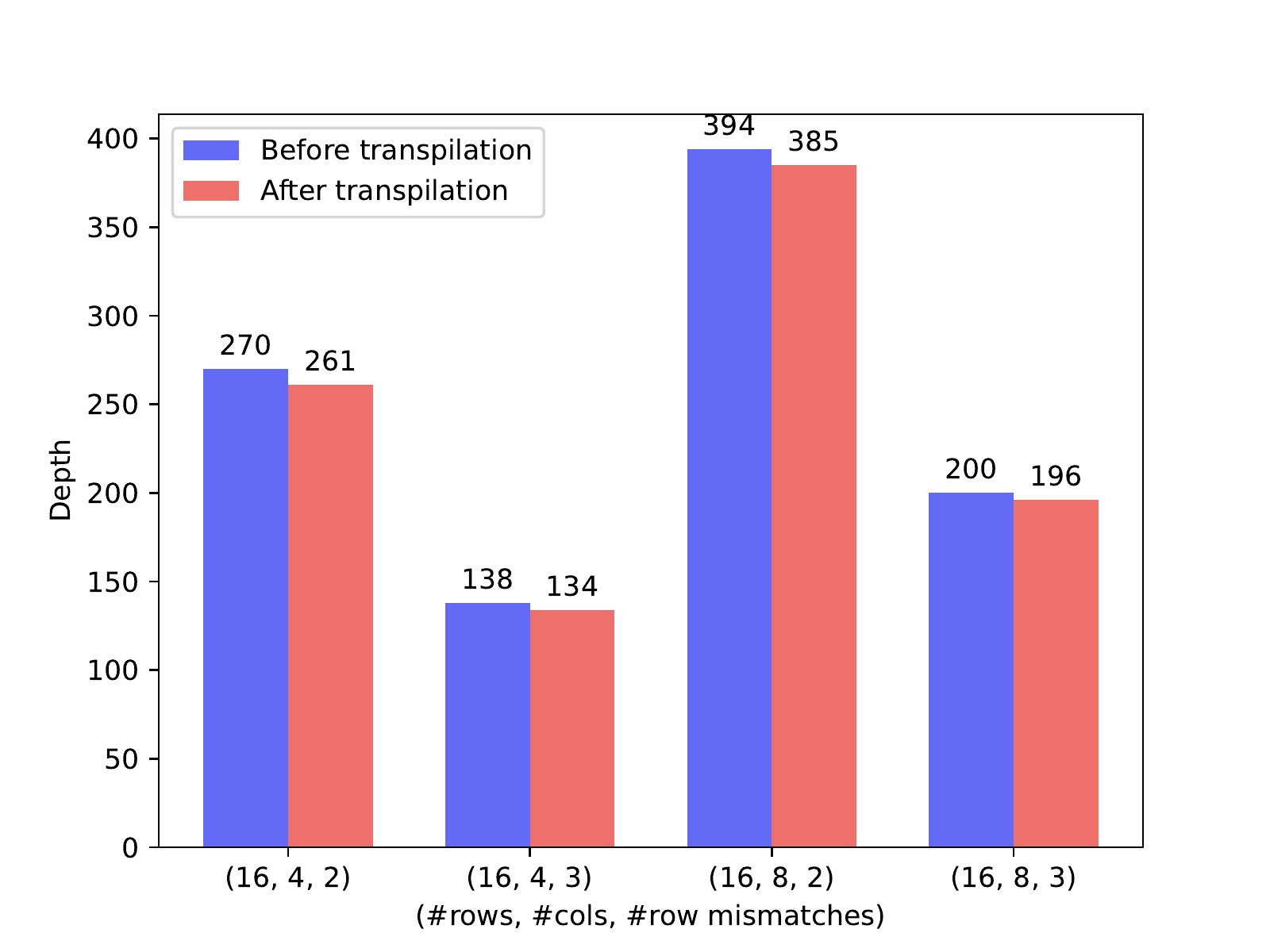}
    \caption{Statevector}
  \end{subfigure}
  \caption{Comparison of circuit depth before/after transpilation on select
  dimensions}
  \label{fig:circuit_depth_before_after}
\end{figure}

\cref{fig:circuit_depth_before_after} shows the circuit depths before and after
transpilation. When using the statevector method we see a small decrease in the
circuit depth. Using MPS, on the other hand, results in two to three orders of
magnitude more gate count and circuit depth. The additional gates for MPS seem
to be mostly CNOTs. A cursory look at a function call profile reveals that
transpilation for MPS spends most of its time performing SWAPs (which consist of
CNOTs). Future work will look into investigating why these SWAPs are performed
even though the target architecture (which is Aer in our case) is
fully-connected.

From these Aer simulations, we conclude that it is possible to run QVMP on up to
32 qubits using NISQ hardware.

\subsection{Transpilation and Simulation}

\begin{figure}[h!]
  \begin{subfigure}{0.48\textwidth}
    \includegraphics[width=\textwidth]{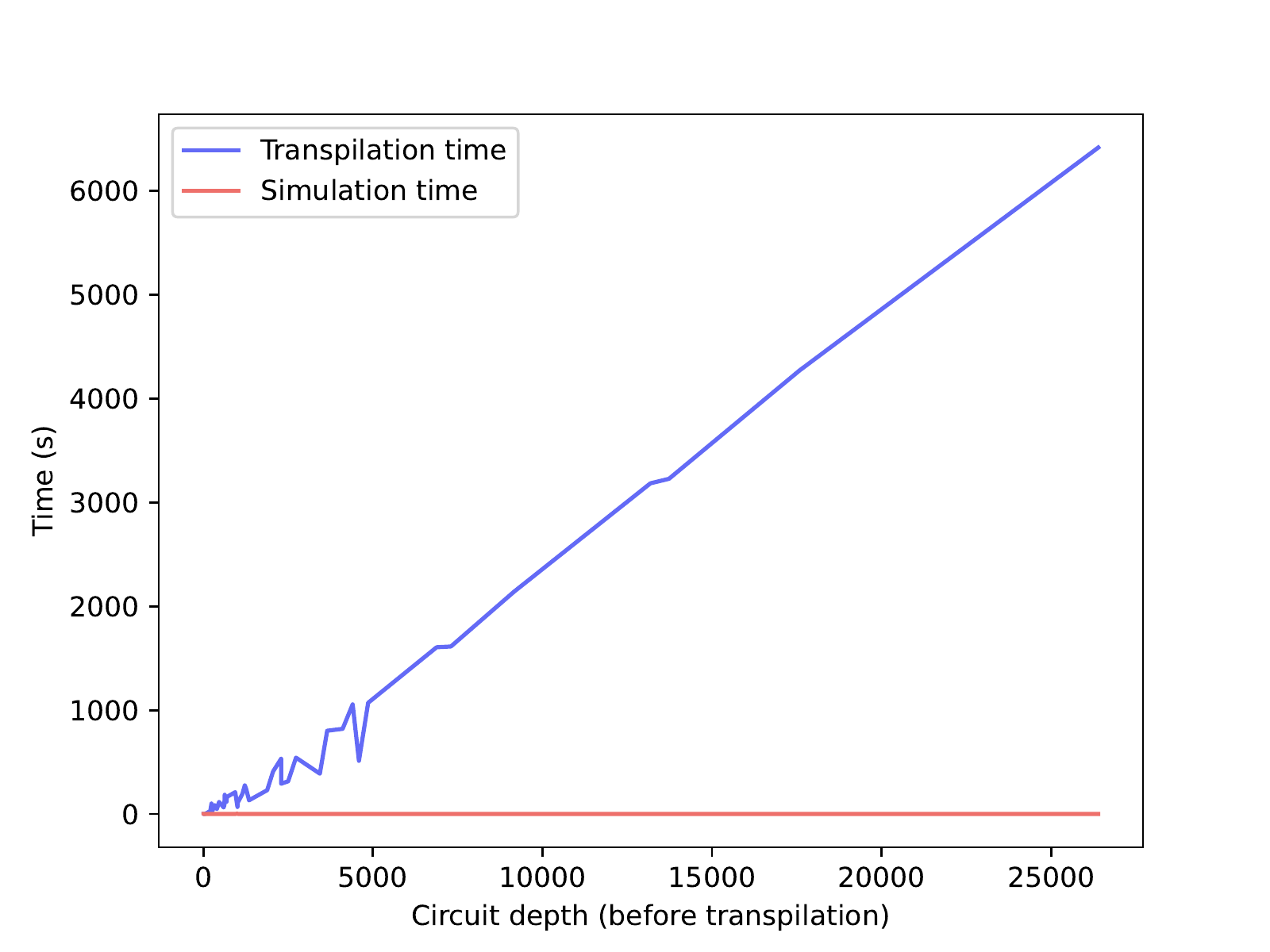}
    \caption{MPS}
  \end{subfigure}
  \begin{subfigure}{0.48\textwidth}
    \includegraphics[width=\textwidth]{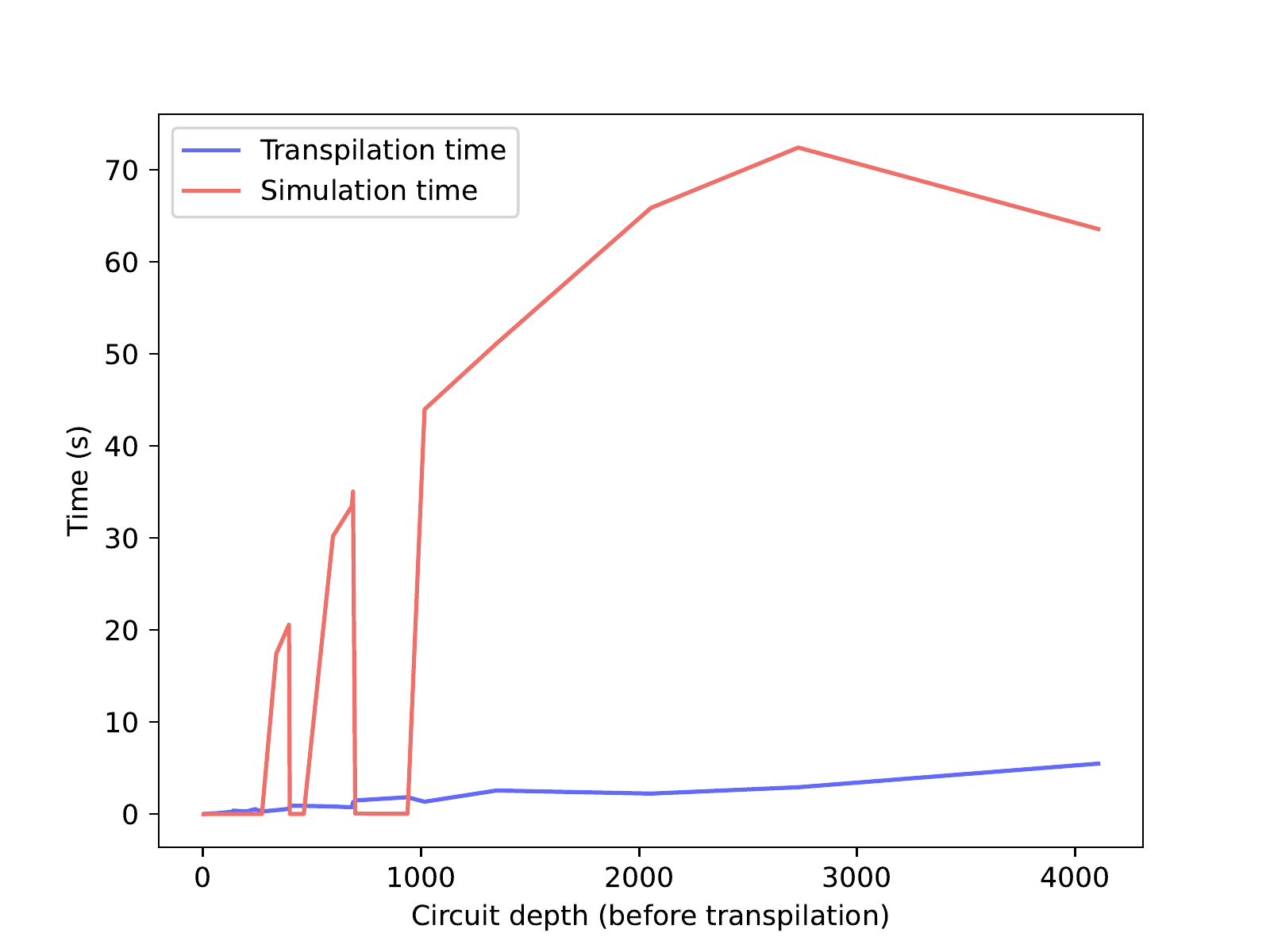}
    \caption{Statevector}
  \end{subfigure}
  \caption{Circuit depth vs Transpilation/Simulation time}
  \label{fig:circuit_depth_of_trans_v_sim}
\end{figure}

\cref{fig:circuit_depth_of_trans_v_sim} showcases how transpilation
time and simulation time change as the circuit depth increases. We observe
opposite trends. MPS seems to be spending more time on transpilation while its
simulation time remains mostly constant. Statevector, on the other hand,
spends more time during simulation with transpilation time growing at a slow
rate. We observe this trend mainly due to the gate explosion described in
\cref{sec:circuit_metrics}.

\subsection{Challenges}

Implementing and debugging Grover oracles in Qiskit is error-prone. The library
does not offer support for encoding matrices into quantum circuits and
performing operations on them.  The choice of encoding scheme is left to the
programmer. This can slow development time since programmers are now forced to
work at lower level of abstraction than desired. Debugging quantum circuits can
be tricky. Consider the following scenario. You've implemented a Grover oracle.
After running the simulation, however, you observe that your circuit is not
amplifying the correct states. This error could be due to a number of reasons:
ancilla qubits may have become entangled with the input qubits or diffusion
might have been applied on the wrong set of qubits, to name a few. Debugging
this in Qiskit involves inserting statevector probes to view the state of the
quantum system and then identify any discrepancies.  While not difficult to do,
it can eat up developer time. We believe that Grover oracle synthesis should be
automated to a high-degree. When possible, such a system should be able to
produce a proof of correctness.

There is also the challenge of scaling quantum circuit development.
Transpilation and simulation times can become bottlenecks in the iteration
cycle.

\section{Conclusion}

In this study we present an implementation of QVMP in Qiskit. We reported
circuit metrics (gate count, depth, qubit count) as well as transpilation and
simulation times. While QVMP can be simulated (and even run) on moderately-sized
inputs, it cannot, with current hardware, scale to a degree where we can
observe any quantum advantage. We demonstrate that the simulation method
selected can have a noticeable impact on depth of the circuit, the effects of
which trickle into transpilation and simulation time.

\begin{figure}
  \begin{lstlisting}[frame=single,language=ML, numbers=left]
(* Example program describing the QVMP oracle *)

[@@oracle]
let find_row_mismatch a y z =
  find_idx (fun idx value -> value <> z[idx]) (a * y)
  \end{lstlisting}
  \caption{High-level description of QVMP in the ML dialect.
  @@oracle is an attribute specifying that the expression that
  follows is an oracle that can be compiled to a reversible circuit}
  \label{fig:high_level_qvmp_oracle}
\end{figure}

Future work in this space involves extending existing reversible compilers (like
REVS \cite{amy2017verified} and a subset of Quipper \cite{green2013quipper}) to
support higher-level programming constructs like lists, records, and
multi-dimensional arrays. This will allow programmers to specify oracles as
high-level classical decision functions without having to think about encoding
schemes, thereby raising the level of abstraction which can boost productivity.
\cref{fig:high_level_qvmp_oracle} demonstrates such a description for QVMP.

Another direction involves investigating the gate and circuit depth
explosion observed when using the MPS simulation method. Cursory
investigation has revealed that the transpiler is inserting multiple SWAP gates
despite the target architecture being fully-connected.

Finally, we also wish to explore the space of more efficient encodings of
matrices and related operations to cut down on qubit counts.

\newpage

\printbibliography[title=Bibliography]

@book{nielsen2000quantum,
	title        = {Quantum computing and quantum information},
	author       = {Nielsen, Michael A and Chuang, Isaac L},
	publisher    = {Cambridge University Press},
	year = {2000},
}

@article{acasiete2020implementation,
  title={Implementation of quantum walks on IBM quantum computers},
  author={Acasiete, Frank and Agostini, Flavia P and Moqadam, J Khatibi and Portugal, Renato},
  journal={Quantum Information Processing},
  volume={19},
  number={12},
  pages={1--20},
  year={2020},
  publisher={Springer},
}

@article{balu2018physical,
  title={Physical realization of topological quantum walks on IBM-Q and beyond},
  author={Balu, Radhakrishnan and Castillo, Daniel and Siopsis, George},
  journal={Quantum Science and Technology},
  volume={3},
  number={3},
  pages={035001},
  year={2018},
  publisher={IOP Publishing},
}

@article{fingerhuth2018open,
	title        = {Open source software in quantum computing},
	author       = {Fingerhuth, Mark and Babej, Tomáš and Wittek, Peter},
	volume       = {13},
	number       = {12},
	pages        = {e0208561},
	doi          = {10.1371/journal.pone.0208561},
	issn         = {1932-6203},
	url          = {https://dx.plos.org/10.1371/journal.pone.0208561},
	journaltitle = {{PLOS} {ONE}},
	shortjournal = {{PLoS} {ONE}},
	editor       = {Mueck, Leonie Anna},
	date         = {2018-12-20},
	langid       = {english},
}

@incollection{jaques2020implementing,
	title        = {Implementing Grover Oracles for Quantum Key Search on {AES} and {LowMC}},
	author       = {Jaques, Samuel and Naehrig, Michael and Roetteler, Martin and Virdia, Fernando},
	booktitle    = {Advances in Cryptology – {EUROCRYPT} 2020},
	location     = {Cham},
	publisher    = {Springer International Publishing},
	volume       = {12106},
	pages        = {280--310},
	doi          = {10.1007/978-3-030-45724-2_10},
	isbn         = {978-3-030-45723-5 978-3-030-45724-2},
	url          = {https://link.springer.com/10.1007/978-3-030-45724-2_10},
	note         = {Series Title: Lecture Notes in Computer Science},
	editor       = {Canteaut, Anne and Ishai, Yuval},
	date         = {2020},
	langid       = {english},
}

@article{larose2019overview,
	title        = {Overview and Comparison of Gate Level Quantum Software Platforms},
	author       = {{LaRose}, Ryan},
	volume       = {3},
	pages        = {130},
	doi          = {10.22331/q-2019-03-25-130},
	issn         = {2521-327X},
	url          = {https://quantum-journal.org/papers/q-2019-03-25-130/},
	journaltitle = {Quantum},
	shortjournal = {Quantum},
	date         = {2019-03-25},
	langid       = {english},
}

@inproceedings{mandviwalla2018implementing,
	title        = {Implementing Grover’s Algorithm on the {IBM} Quantum Computers},
	author       = {Mandviwalla, Aamir and Ohshiro, Keita and Ji, Bo},
	booktitle    = {2018 {IEEE} International Conference on Big Data (Big Data)},
	location     = {Seattle, {WA}, {USA}},
	publisher    = {{IEEE}},
	pages        = {2531--2537},
	doi          = {10.1109/BigData.2018.8622457},
	isbn         = {978-1-5386-5035-6},
	url          = {https://ieeexplore.ieee.org/document/8622457/},
	eventtitle   = {2018 {IEEE} International Conference on Big Data (Big Data)},
	date         = {2018-12}
}

@incollection{dustdar2020nisq,
	title        = {The {NISQ} Analyzer: Automating the Selection of Quantum Computers for Quantum Algorithms},
	shorttitle   = {The {NISQ} Analyzer},
	author       = {Salm, Marie and Barzen, Johanna and Breitenbücher, Uwe and Leymann, Frank and Weder, Benjamin and Wild, Karoline},
	booktitle    = {Service-Oriented Computing},
	location     = {Cham},
	publisher    = {Springer International Publishing},
	volume       = {1310},
	pages        = {66--85},
	doi          = {10.1007/978-3-030-64846-6_5},
	isbn         = {978-3-030-64845-9 978-3-030-64846-6},
	url          = {https://link.springer.com/10.1007/978-3-030-64846-6_5},
	note         = {Series Title: Communications in Computer and Information Science},
	editor       = {Dustdar, Schahram},
	date         = {2020},
	langid       = {english},
}

@article{ambainis2007quantumwalk,
	title        = {Quantum Walk Algorithm for Element Distinctness},
	author       = {Ambainis, Andris},
	volume       = {37},
	number       = {1},
	pages        = {210--239},
	doi          = {10.1137/S0097539705447311},
	issn         = {0097-5397, 1095-7111},
	url          = {http://epubs.siam.org/doi/10.1137/S0097539705447311},
	journaltitle = {{SIAM} Journal on Computing},
	shortjournal = {{SIAM} J. Comput.},
	date         = {2007-01},
	langid       = {english},
}

@article{ambainis2002quantummatrix,
	title        = {Quantum matrix verification},
	author       = {Ambainis, Andris and Buhrman, Harry and Høyer, Peter and Karpinski, Marek and Kurur, P},
	journaltitle = {Unpublished manuscript},
	date         = {2002}
}

@incollection{freivalds1979fast,
	title        = {Fast probabilistic algorithms},
	author       = {Freivalds, Rūsiņš},
	booktitle    = {International Symposium on Mathematical Foundations of Computer Science},
	publisher    = {Springer},
	pages        = {57--69},
	date         = {1979}
}

@article{mintz2020qcor,
	title        = {{QCOR}: A Language Extension Specification for the Heterogeneous Quantum-Classical Model of Computation},
	shorttitle   = {{QCOR}},
	author       = {Mintz, Tiffany M. and {McCaskey}, Alexander J. and Dumitrescu, Eugene F. and Moore, Shirley V. and Powers, Sarah and Lougovski, Pavel},
	volume       = {16},
	number       = {2},
	pages        = {1--17},
	doi          = {10.1145/3380964},
	issn         = {1550-4832, 1550-4840},
	url          = {https://dl.acm.org/doi/10.1145/3380964},
	journaltitle = {{ACM} Journal on Emerging Technologies in Computing Systems},
	shortjournal = {J. Emerg. Technol. Comput. Syst.},
	date         = {2020-04-30},
	langid       = {english}
}

@article{lanl2018quantum,
	title        = {Quantum algorithm implementations for beginners},
	author       = {Adedoyin, Adetokunbo and Ambrosiano, John and Anisimov, Petr and B{\"a}rtschi, Andreas and Casper, William and Chennupati, Gopinath and Coffrin, Carleton and Djidjev, Hristo and Gunter, David and Karra, Satish and others},
	year         = {2018},
	journal      = {arXiv preprint arXiv:1804.03719}
}

@inproceedings{amy2017verified,
	title        = {Verified compilation of space-efficient reversible circuits},
	author       = {Amy, Matthew and Roetteler, Martin and Svore, Krysta M},
	year         = {2017},
	booktitle    = {International Conference on Computer Aided Verification},
	pages        = {3--21},
	organization = {Springer}
}

@article{babbush2018encoding,
	title        = {Encoding electronic spectra in quantum circuits with linear T complexity},
	author       = {Babbush, Ryan and Gidney, Craig and Berry, Dominic W and Wiebe, Nathan and McClean, Jarrod and Paler, Alexandru and Fowler, Austin and Neven, Hartmut},
	year         = {2018},
	journal      = {Physical Review X},
	publisher    = {APS},
	volume       = {8},
	number       = {4},
	pages        = {041015}
}

@article{buhrman2005quantum,
	title = {Quantum Verification of Matrix Products},
	url = {http://arxiv.org/abs/quant-ph/0409035},
	journaltitle = {{arXiv}:quant-ph/0409035},
	author = {Buhrman, Harry and Spalek, Robert},
	date = {2005-07-06},
	eprinttype = {arxiv},
	eprint = {quant-ph/0409035},
}

@inproceedings{green2013quipper,
  title={Quipper: a scalable quantum programming language},
  author={Green, Alexander S and Lumsdaine, Peter LeFanu and Ross, Neil J and Selinger, Peter and Valiron, Beno{\^\i}t},
  booktitle={Proceedings of the 34th ACM SIGPLAN conference on Programming language design and implementation},
  pages={333--342},
  year={2013}
}

\end{document}